\documentclass[11pt]{article}
\usepackage[utf8x]{inputenc}
\usepackage{graphicx}
\usepackage{float}
\usepackage{latexsym}
\usepackage{amsmath}
\usepackage{amssymb}
\usepackage{amsthm}
\usepackage{amsfonts}
\usepackage{color}

\numberwithin{equation}{section}
\newtheorem{definition}{Definition}[section]
\newtheorem{theorem}{Theorem}[section]
\newtheorem{lemma}{Lemma}[section]
\theoremstyle{definition}

\newtheorem{remark}{Remark}[section]
\newtheorem{corollary}{Corollary}[section]

\newcommand{\D}{\displaystyle}

\newcommand{\e}{\epsilon}

\newcommand{\mathd}{\mathrm{d}}

\title{On the Uniqueness of Sparse Time-Frequency Representation of Multiscale Data}
\author{Chunguang Liu \thanks{Department of Mathematics, Jinan University, Guangzhou, China, 510632.
{\it Email: tcgliu@jnu.edu.cn.}} \and Zuoqiang Shi \thanks{Mathematical Sciences Center, Tsinghua University, Beijing, China, 100084.
{\it Email: zqshi@math.tsinghua.edu.cn.}} \and Thomas Y. Hou \thanks{Applied and Comput. Math, MC 9-94, Caltech,
Pasadena, CA 91125. {\it Email: hou@cms.caltech.edu.}}}

\begin{document}

\maketitle

\abstract{In this paper, we analyze the uniqueness of the sparse time frequency decomposition and investigate the
efficiency of the nonlinear matching pursuit method. Under the assumption of scale separation, we show that the sparse
time frequency decomposition is unique up to an error that is determined by the scale separation property of the signal. We further show that the unique decomposition can be obtained approximately by the sparse time frequency decomposition using nonlinear matching pursuit.}


\section{Introduction}

Nowadays, data play more and more important role in our life.
It has become increasingly important to develop effective data analysis tools to extract useful information from massive amount of data.
Frequency is one of the most important features for oscillatory data. In many physical problems, frequencies encode important information
of the underlying physical mechanism. Many time-frequency analysis methods have been developed to extract information of frequencies and the corresponding amplitudes from the measurement of signals. These include the windowed Fourier transform, the wavelet transform
\cite{Daub92, Mallat09}, the Wigner-Ville distribution \cite{Flandrin99}, and the Empirical Mode Decomposition (EMD) method \cite{EMD02,EMD09}. 
Among these different time-frequency analysis methods, the EMD method provides an efficient adaptive method to extract frequency information from nonlinear and nonstationary data and it has been successfully used in 
many applications. However, due to its empirical nature, the EMD method still lacks a rigorous mathematical foundation. 
Recently, a number of attempts have been made to provide a mathematical foundation for this method,
 see e.g. the  synchrosqueezed wavelet transform \cite{DLW11}, the Empirical wavelet transform \cite{DZ13}, the variational mode decomposition \cite{Gilles13}.

In the last few years, inspired by the EMD method and compressive sensing \cite{CRT06,CT06,Dnh06},
 Hou and Shi proposed a novel time frequency analysis method based on the sparsest representation of multiscale data \cite{HS13}. 
In this method, the signal is decomposed into a finite number of intrinsic mode functions with a small residual: 
\begin{eqnarray}
\label{decomp-f}
  f(t)=\sum_{j=1}^M a_j(t)\cos\theta_j(t) + r(t),\quad t\in \mathbb{R} ,
\end{eqnarray}
where $a_j(t),\,\theta_j'(t)>0,\;j=1,\cdots,M$ and $r(t)$ is a small residual. We assume that
$a_j(t)$ and $\theta_j'$ are less oscillatory than $\cos\theta_j(t)$. What we mean by ``less oscillatory'' will be made precise below by the definition of scale separation property of the signal.
Borrowing the terminology from the EMD method, we call $a_j(t)\cos\theta_j(t)$ the Intrinsic Mode Function (IMF) \cite{EMD02}.
After the decomposition is obtained, the instantaneous frequencies $\omega_j(t)$ are defined as
\begin{eqnarray}
  \omega_j(t)=\theta_j'(t),
\end{eqnarray}
and the amplitude is given by $a_j(t)$.

One main difficulty in computing the decomposition \eqref{decomp-f} is that the decomposition is not unique.
To pick up the "best" decomposition among all feasible ones, Hou and Shi proposed to decompose the signal by looking for the sparsest decomposition
by solving the following nonlinear optimization problem:
 \begin{eqnarray}
\label{opt-l0-pre}
\quad && \begin{array}{rcc}\vspace{-2mm}
&\mbox{Minimize} &M\\ \vspace{2mm}
&{\scriptstyle (a_k)_{1\le k\le M}, (\theta_k)_{1\le k\le M}}&\\
&\mbox{Subject to:}&\D \|f-\sum_{k=1}^M a_k\cos\theta_k \|_{L^2} \le \epsilon,\quad\quad a_k\cos\theta_k\in \mathcal{D},
\end{array}
\end{eqnarray}
where $\epsilon$ is the noise level and $\mathcal{D}$ is the dictionary consist of all IMFs (see \cite{HS13} for its precise definition).

The idea of looking for the sparsest representation over the time frequency dictionary has been exploited extensively in the signal processing community, see, e.g. \cite{MZ93,CDS98}. Comparing with other existing methods, the novelty of the method proposed by Hou and Shi is that the time frequency dictionary being used is much larger. This method has the advantage of being fully adaptive to the signal.  


The optimization problem \eqref{opt-l0-pre} is nonlinear and nonconvex. It is challenging to solve this nonlinear optimization problem efficiently.
To overcome this difficulty, Hou and Shi proposed an efficient algorithm based on nonlinear matching pursuit and Fast Fourier transform to solve the above nonlinear optimization problem.
In a subsequent paper \cite{HST13-2}, the authors proved the convergence of their nonlinear matching pursuit algorithm for periodic data
that satisfy certain scale separation property. Further, they have demonstrated the effectiveness of this method by decomposing realistic signals arising from various applications. However, from the theoretical point of view, one important question remains open. That is the uniqueness of the solution of the optimization problem \eqref{opt-l0-pre}. This is precisely the main focus of this paper.

In this paper, we will show that under the assumption of scale separation, the solution of optimization problem \eqref{opt-l0-pre} is unique
up to an error determined by the scale separation property. 
First, we give a precise definition of scale separation of a signal as follows:
\begin{definition}
  [scale-separation]
\label{scale-seperation}
One function $f(t)=a(t)\cos\theta(t)$ is said to satisfy a
scale-separation property with a separation factor $\e >0$,
if $a(t)$ and $\theta(t)$ satisfy the following conditions:
\begin{eqnarray*}
 && a(t)\in C^1(\mathbb{R}),\; \theta\in C^2(\mathbb{R}),\quad
\inf_{t\in \mathbb{R}} \theta'(t)>0,
\\
&& \frac{\sup_{t\in \mathbb{R}}\theta'(t)}{\inf_{t\in \mathbb{R}}\theta'(t)}= M'<+\infty,
\quad \left|\frac{a'(t)}{\theta'(t)}\right|\le \e,\; \;
\left|\frac{\theta''(t)}{\left(\theta'(t)\right)^2}\right|\le \e,\quad \forall t\in
\mathbb{R}.
\end{eqnarray*}
\end{definition}
In the above definition, the first three assumptions are on the regularity of the envelope $a(t)$ and the instantaneous frequency $\theta'(t)$. This regularity is relatively mild and can be relaxed to a piecewise smooth function. The key assumptions on the scale separation property are the last two assumptions. The assumption $\left|\frac{a'(t)}{\theta'(t)}\right|\le \e$ quantifies what we mean by the envelope is less oscillatory than the normalized signal $\cos(\theta(t))$. The last assumption 
$\left|\frac{\theta''(t)}{\left(\theta'(t)\right)^2}\right|\le \e$
essentially says that 
the frequency oscillation is relatively weak compared with the square of the frequency itself. In the previous theoretical study of time-frequency analysis, a much stronger assumption on the frequency oscillation is made, i.e. $\left|\frac{\theta''}{\theta'} \right|\le \e$, see e.g. \cite{DLW11,CM15}. 
In the problems that we consider, the instantaneous frequency $\theta'(t)$ is typically quite large. Therefore, the assumption $\left|\frac{\theta''(t)}{\left(\theta'(t)\right)^2}\right|\le \e$ is
 much weaker than the typical assumption, $\left|\frac{\theta''}{\theta'} \right|\le \e$.

With the definition of scale separation, we construct the dictionary $\mathcal{D}_\e$ by putting all the functions satisfying 
scale separation together. 
\begin{align}\label{dict1}
 \mathcal{D}_\e:=\left\{a(\cdot)\cos\theta(\cdot):
(a, \theta)\in U_\e\right\},
\end{align}
where
\begin{align}\label{opt03}
 U_\e:=\left\{(a, \theta): a>0, \theta'>0; \; \left|\frac{a'}{\theta'}\right|\le\e,\;
\left|\frac{\theta''}{[\theta']^2}\right|\le\e,\;\frac{\sup_{t\in\mathbb{R}}\theta'(t)}{\inf_{t\in\mathbb{R}}\theta'(t)}
=M' <+\infty \right\}.
\end{align}
$\mathcal{D}_\e$ is the dictionary that we will use to represent the signal. 
For any given signal $f(t)$, we decompose $f$ over the dictionary $\mathcal{D}_\e$ by looking for the sparsest representation:
\begin{align}\label{P0}\tag{$P_0$}
\begin{array}{rcl}
 \mbox{Minimize} & &\quad\quad M \\
 \mbox{subject to} & & \D |f(t)-\sum_{k=1}^M a_k(t)\cos\theta_k(t)|\le \e_0, \\
&&a_k(t)\cos\theta_k(t)
\in \mathcal{D}_\e, \quad k=1,2,\cdots, M.
\end{array}
\end{align}
Here $\e_0$ is a given threshold of the accuracy of the decomposition. Typically, $\e_0$ is set according to the amplitude of noise.

To get the uniqueness, we need to assume that the signal $f(t)$ is well separated which is defined in Definition \ref{def:well-sep-intro}.
\begin{definition}
  [Well-separated signal]
\label{def:well-sep-intro}
A signal $f: \mathbb{R}\rightarrow \mathbb{R}$ is said to be
well-separated with separation factor $\e$ and frequency
ratio $d$ if it can be written as
\begin{eqnarray}
\label{opt01}
  f(t)=\sum_{k=1}^Ma_k(t)\cos\theta_k(t)+r(t)
\end{eqnarray}
where all $f_k(t)=a_k(t)\cos\theta_k(t)$ satisfies the
scale-separation property with separation factor $\e$, $r(t)=O(\e_0)$ and
their phase functions $\theta_k$ satisfy
\begin{eqnarray}
  \label{seperation-IMF-intro}
\theta_k'(t)\ge d \theta_{k-1}'(t),\quad \forall t\in \mathbb{R},
\end{eqnarray}
and $d>1$, $d-1=O(1)$.
\end{definition}
In the above definition, the key measurement of separation among different components is the value $d$, which measures the frequency gap among different components.  For a well separated signal, we can prove that the solution of the optimization problem \eqref{P0} is unique up to an error $O(\epsilon)$. The main result is summarized in Theorem \ref{thm:unique-intro}.
\begin{theorem}
  \label{thm:unique-intro}
  Let $f(t)$ be well separated with
separation factor $\e\ll 1$ and frequency ratio $d$ as defined in
Definition \ref{def:well-sep-intro}.
Then $\left(a_k,\theta_k\right)_{1\le k\le M}$ is an optimal solution of the optimization problem \eqref{P0} and it is unique up to the error $\e$, i.e.
if $\left(\tilde{a}_k,\tilde{\theta}_k\right)_{1\le k\le \tilde{M}}$ is another optimal solution of \eqref{P0}, then $\tilde{M}=M$ and
\begin{eqnarray}
\label{eqn:unique-error}
  |a_k(t)-\tilde{a}_k(t)|=O(\e),\quad \frac{|\theta_k(t)-\tilde{\theta}_k(t)|}{\theta'_k(t)}=O(\e), \quad \forall t,\quad k=1,\cdots,M .
\end{eqnarray}
\end{theorem}
This theorem is proved by carefully studying the wavelet transform of each IMF. 
The details of the proof can be found in Section 2.


We remark that there has been some very nice progress in developing
a mathematical framework for an EMD like method using synchrosqueezed
wavelet transforms \cite{DLW11} and windowed Fourier transform \cite{CM15}. 
 These methods
are based on continuous wavelet transform and windowed Fourier transform respectively and do not look for the sparsest decomposition directly. 
So the question of uniqueness is not the same as that we consider here in this paper.

The rest of the paper is organized as follows. The uniqueness of the sparse time-frequency decomposition is analyzed
in Section 2. In Section 3, we analyze the performance of an algorithm based on matching pursuit. Some concluding remarks are made in Section 4. We defer a few technical lemmas to the appendices.

\section{Uniqueness of \ref{P0} for well-separated signals}
\label{sec:unique}

In this section, we assume that the signal, $f(t)$, satisfies the scale-separation property with
separation factor $\e$ and frequency ratio $d$ as defined in
Definition \ref{def:well-sep-intro}.

For this kind of signals, the existence of the solution of \eqref{P0} is obvious. Since $f(t)$ already has a representation,
\begin{eqnarray}
  f(t)=\sum_{k=1}^Ma_k(t)\cos\theta_k(t)+r(t),
\end{eqnarray}
this gives a feasible decomposition. 
Each feasible decomposition gives a positive integer, by collecting all these positive integers together, we get a set $A$. Then we know that
$A$ is nonempty and has a lower bound. Let $M_0=\inf A$. Since $A$ consists of positive integers, then $M_0=\inf A$ can be achieved. Then
the existence of the solution of \eqref{P0} is proved.

But the uniqueness
is much more complicated. In the next subsection, we will prove that if $f$ 
is a well-separated signal with separation factor $\e$ and frequency
ratio $d>1$ and $\e\ll 1$, then
 the solution of
\eqref{P0} is unique up to $\e$ and $\e_0$.

To simplify the notation, in the rest of this paper, we assume $\e_0$ and $\e$ have the same order and denote both of them by $\e$.

%

We prove Theorem \ref{thm:unique-intro} by carefully studying the wavelet transform of each IMF. As we show in Lemma \ref{lem:wft}, the continuous wavelet transform (CWT) of each IMF is confined in a narrow band by properly choosing an appropriate wavelet function. Then the uniqueness can be obtained by comparing the continuous wavelet transform of different decompositions.

In order to complete the proof, first, we need to estimate the width of CWT of each IMF. This estimation is given in the following lemma.
\begin{lemma}
\label{lem:wft}
let $\psi$ is a wavelet function
such that
\begin{eqnarray*}
  I_1=\int_{\mathbb{R}} |\psi(\tau)|\mathd \tau<+\infty,\quad I_2=\int_{\mathbb{R}}|\tau\psi'(\tau)|\mathd \tau<+\infty,\quad I_3=\int_{\mathbb{R}}|\tau^2\psi''(\tau)|\mathd \tau<+\infty.
\end{eqnarray*} 
Suppose $(a,\theta)\in U_\e$, i.e.
\begin{eqnarray*}
  \left|\frac{a'}{\theta'}\right|\le\e,\quad
\left|\frac{\theta''}{[\theta']^2}\right|\le\e,\quad \frac{\sup_{t\in\mathbb{R}}\theta'(t)}{\inf_{t\in\mathbb{R}}\theta'(t)}= M'<+\infty
\end{eqnarray*}
then we have
  \begin{eqnarray*}
\mathcal{W}(ae^{-i\theta})(t,\omega)=  
\frac{1}{\sqrt{\omega}}\int_{\mathbb{R}} a(\tau)e^{-i\theta(\tau)}\psi\left(\frac{\tau-t}{\omega}\right)\mathd \tau
=\sqrt{\omega}a(t)e^{-i\theta(t)}\widehat{\psi}(\omega\theta'(t))
+C\sqrt{\omega}\,\e
\end{eqnarray*}
where $C=(A+4|a(t)|+1)I_1+\left[M'+(M'+1)|a(t)|\right]I_2+M'|a(t)|I_3$  and $A=\sup_{t\in \mathbb{R}}|a(t)|$.

\end{lemma}
The proof of this lemma can be found in Appendix A.
\begin{remark}
One sufficient condition to make sure that $I_1,I_2,I_3<+\infty$ is that
the Fourier transform of the wavelet, $\widehat{\psi}\in C^4(\mathbb{R})$ and has compact support. In the proof of Theorem \ref{thm:unique-intro}, 
we use the wavelet whose Fourier transform is a fifth order B-spline function such that above sufficient condition is satisfied.
\end{remark}

\subsection{Proof of Theorem \ref{thm:unique-intro}}

Now, we are ready to give the proof of Theorem \ref{thm:unique-intro}.
\begin{proof} {\it of Theorem \ref{thm:unique-intro}}

The details of the proof may be a bit tedious but the idea is very clear. First, we assume there are two decompositions:
\begin{eqnarray}
  f(t)=\sum_{k=1}^Ma_k(t)\cos\theta_k(t)+O(\e)=\sum_{k=1}^{\tilde{M}}\tilde{a}_k(t)\cos\tilde{\theta}_k(t)+O(\e) .
\end{eqnarray}
Using Lemma \ref{lem:wft}, we have
\begin{eqnarray}
\label{eqn:wft-1}
  \omega^{-1/2}\mathcal{W}(f)(t,\omega)=\sum_{k=1}^Ma_k(t)e^{-i\theta_k(t)}\widehat{\psi}(\omega\theta'_k(t))+O(\e)=
\sum_{k=1}^{\tilde{M}}\tilde{a}_k(t)e^{-i\tilde{\theta}_k(t)}\widehat{\psi}(\omega\tilde{\theta}'_k(t))+O(\e).
\end{eqnarray}
By analyzing the support of $\mathcal{W}(f)(t,\omega)$, we can get $M=\tilde{M}$ and
\begin{eqnarray}
  a_k(t)e^{-i\theta_k(t)}\widehat{\psi}(\omega\theta'_k(t))=\tilde{a}_k(t)e^{-i\tilde{\theta}_k(t)}\widehat{\psi}(\omega\tilde{\theta}'_k(t))+O(\e),\quad
k=1,\cdots,M.
\end{eqnarray}
Using these equalities, we can show that \eqref{eqn:unique-error} is true.

First, we pick up a specific wavelet function $\psi$. We require that its Fourier transform $\widehat{\psi}\in C^4$ has support in
$[1-\Delta, 1+\Delta]$ with $0<\Delta < \frac{\sqrt{d}-1}{\sqrt{d}+1}$ and
$\widehat{\psi}(1)=1$ is the maximum of $|\widehat{\psi}|$. In this proof, we choose $\widehat{\psi}$ to be a fifth order B-spline function after proper scaling
and translation.

For any $\theta\in C^1$, define
  \begin{eqnarray*}
    U_\theta=\left\{(t,\omega)\in \mathbb{R}^2: |\widehat{\psi}(\omega \theta'(t))|> \e\right\},\quad
U_\theta(t)=\left\{\omega\in \mathbb{R}: |\widehat{\psi}(\omega \theta'(t))|> \e\right\}.
  \end{eqnarray*}

Now, fix a time $t=t_0\in [0,1]$, for any $l\in \{1,\cdots,M\}$, let $\omega_{l,0}=\frac{1}{\theta_l'(t_0)}$. Using \eqref{eqn:wft-1}, we have
\begin{eqnarray}
  \omega^{-1/2}  \mathcal{W}(f)(t_0,\omega_{l,0})&=&\sum_{k=1}^Ma_k(t_0)e^{-i\theta_k(t_0)}\widehat{\psi}(\theta'_k(t_0)/\theta'_l(t_0))+O(\e)\nonumber\\
&=&\sum_{k=1}^{\tilde{M}}\tilde{a}_k(t_0)e^{-i\tilde{\theta}_k(t_0)}\widehat{\psi}(\tilde{\theta}'_k(t_0)/\theta'_l(t_0))+O(\e).
\end{eqnarray}
Since $\theta'_k(t_0)/\theta'_l(t_0)\le 1/d<1-\Delta$ or $\theta'_k(t_0)/\theta'_l(t_0)\ge d > 1+\Delta$ for any $k\ne l$, in the
first summation of the above equation, we have only one term left,
\begin{eqnarray}
 a_l(t_0)e^{-i\theta_l(t_0)}\widehat{\psi}(1)=\sum_{k=1}^{\tilde{M}}\tilde{a}_k(t_0)e^{-i\tilde{\theta}_k(t_0)}\widehat{\psi}(\tilde{\theta}'_k(t_0)/\theta'_l(t_0))+O(\e).
\end{eqnarray}
Then, there exists at least a $I(l,t_0)\in \{1,\cdots,\tilde{M}\}$, such that
\begin{eqnarray}
\left|\widehat{\psi}(\tilde{\theta}'_{I(l,t_0)}(t_0)/\theta'_l(t_0))\right|>0,
\end{eqnarray}
which means that $1-\Delta< \tilde{\theta}'_{I(l,t_0)}(t_0)/\theta_l'(t_0)<1+\Delta$.
Using the assumption that $f$ is well-separated with frequency ratio $d$ and $0<\Delta < \frac{\sqrt{d}-1}{\sqrt{d}+1}$,
for any $k\ne l$, we obtain
\begin{eqnarray}
  \frac{\theta'_k(t_0)}{\theta_l'(t_0)}\ge d, \quad \text{or}\quad \frac{\theta'_k(t_0)}{\theta_l'(t_0)}\le \frac{1}{d} .
\end{eqnarray}
This gives that
\begin{eqnarray}
  \frac{\tilde{\theta}'_{I(l,t_0)}(t_0)}{\theta'_k(t_0)}=  \frac{\tilde{\theta}'_{I(l,t_0)}(t_0)}{\theta'_l(t_0)}  \cdot
\frac{\theta'_{l}(t_0)}{\theta'_k(t_0)} \ge d(1-\Delta)>1+\Delta ,
\end{eqnarray}
or
\begin{eqnarray}
  \frac{\tilde{\theta}'_{I(l,t_0)}(t_0)}{\theta'_k(t_0)}=  \frac{\tilde{\theta}'_{I(l,t_0)}(t_0)}{\theta'_l(t_0)}  \cdot
\frac{\theta'_{l}(t_0)}{\theta'_k(t_0)} \le \frac{1+\Delta}{d}<1-\Delta .
\end{eqnarray}
Then, for any $k\ne l$, we have
\begin{eqnarray}
\label{eqn:cover}
\left|\widehat{\psi}(\tilde{\theta}'_{I(l,t_0)}(t_0)/\theta'_k(t_0))\right|=0 .
\end{eqnarray}
This implies that $I(k,t_0)\ne I(l,t_0),\; k\ne l$. Then we get that
\begin{eqnarray}
  \tilde{M}\ge M.
\end{eqnarray}
Since $\left(\tilde{a}_k,\tilde{\theta}_k\right)_{1\le k\le \tilde{M}}$ is a solution of \eqref{P0}, we also have $\tilde{M}\le M$. This implies that
 \begin{eqnarray}
\label{eqn:sparsity}
  \tilde{M}= M,
\end{eqnarray}
and for any $t\in [0,1]$, $I(\cdot,t): \{1,\cdots,M\}\rightarrow \{1,\cdots,M\}$ is a one to one map. Then, we can define its inverse map
$I^{-1}(\cdot,t): \{1,\cdots,M\}\rightarrow \{1,\cdots,M\}$.

For any $k=1,\cdots, M$,  we study the function $I^{-1}(k,\cdot): [0,1]\rightarrow \{1,\cdots,M\}$.
Using the condition that $\frac{\tilde{\theta}''}{(\tilde{\theta}')^2}\le \e$ and the signal $f$ is well-separated, it is easy to see that
$I^{-1}(k,\cdot)$ is a constant over $[0,1]$, i.e.
\begin{eqnarray}
  I^{-1}(k,t)=I^{-1}(k,0),\quad \forall t\in [0,1],\; k=1,\cdots,M.
\end{eqnarray}
Otherwise, suppose there exists $t_0\in [0,1]$, such that $I^{-1}(k,t_0)\ne I^{-1}(k,0)$. Let $A=\{0\le t \le t_0: I^{-1}(k,t_0)= I^{-1}(k,0)\}$
and $\xi =\sup A$. Then for any $\eta>0$, there exist $t_1,t_2 \in [0,1]$, such that
\begin{eqnarray*}
  t_1<\xi<t_2,\quad |t_2-t_1|<\eta,\quad I^{-1}(k,0)= I^{-1}(k,t_1)\ne I^{-1}(k,t_2).
\end{eqnarray*}
Denote
\begin{eqnarray}
  I^{-1}(k,t_1)=k_1,\quad   I^{-1}(k,t_2)=k_2,\quad k_1\ne k_2 .
\end{eqnarray}
Then, by the definition of the map $I(k,t)$, we have
\begin{eqnarray*}
  1-\Delta<\frac{\tilde{\theta}'_{k}(t_1)}{\theta'_{k_1}(t_1)}<1+\Delta,\quad   1-\Delta<\frac{\tilde{\theta}'_{k}(t_2)}{\theta'_{k_2}(t_2)}<1+\Delta.
\end{eqnarray*}
Without loss of generality, we assume that $\theta'_{k_2}> \theta'_{k_1}$.
Then, we have
\begin{eqnarray*}
  \tilde{\theta}'_{k}(t_1)<(1+\Delta) \theta'_{k_1}(t_1),\quad \tilde{\theta}'_{k}(t_2)>(1-\Delta) \theta'_{k_2}(t_1) .
\end{eqnarray*}
Now, let $\eta\rightarrow 0$, have $t_1,t_2\rightarrow \xi$, which gives
\begin{eqnarray*}
  \tilde{\theta}'_{k}(\xi)\le(1+\Delta) \theta'_{k_1}(\xi),\quad \tilde{\theta}'_{k}(\xi)\ge(1-\Delta) \theta'_{k_2}(\xi).
\end{eqnarray*}
On the other hand, since $f$ is well-separated, we know that
\begin{eqnarray*}
\frac{\theta'_{k_1}(t)}{\theta'_{k_2}(t)}\le 1/d.
\end{eqnarray*}
Then, we have
\begin{eqnarray*}
  \tilde{\theta}'_{k}(\xi)\le(1+\Delta) \theta'_{k_1}(\xi),\quad \tilde{\theta}'_{k}(\xi)\ge d(1-\Delta) \theta'_{k_1}(\xi)>(1+\Delta)\theta'_{k_1}(\xi),
\end{eqnarray*}
which is a contradiction.
This means that $I^{-1}(k,\cdot)$ is a constant over $[0,1]$ and we can assume 
\begin{eqnarray}
  I^{-1}(k,t)=k,\quad \forall t\in [0,1],\; k=1,\cdots,M.
\end{eqnarray}

Now we have that
\begin{eqnarray}
  1-\Delta<\frac{\tilde{\theta}'_k(t)}{\theta'_k(t)}<1+\Delta, \quad \forall t\in [0,1],\;  k=1,\cdots,M.
\end{eqnarray}
Using the assumption that the signal $f$ is well-separated with ratio $d$ and the choice of $\psi$, we have that for any $k,\;l=1,\cdots,M,\; k\ne l$
\begin{eqnarray}
\label{eqn:nojoint-1}
\widehat{\psi}\left(\omega\theta'_l(t)\right)=\widehat{\psi}\left(\omega\tilde{\theta}'_l(t)\right)=0,\quad \forall (t,\omega)\in U_{\theta_k},\\
\label{eqn:nojoint-2}
\widehat{\psi}\left(\omega\theta'_l(t)\right)=\widehat{\psi}\left(\omega\tilde{\theta}'_l(t)\right)=0,\quad \forall (t,\omega)\in U_{\tilde{\theta}_k}.
\end{eqnarray}
On the other hand, using Lemma \ref{lem:wft} we have
\begin{eqnarray}
\label{eqn:wft}
    \omega^{-1/2}\mathcal{W}(f)(t,\omega)=\sum_{k=1}^Ma_k(t)e^{-i\theta_k(t)}\widehat{\psi}(\omega\theta'_k(t))+O(\e)=
\sum_{k=1}^M\tilde{a}_k(t)e^{-i\tilde{\theta}_k(t)}\widehat{\psi}(\omega\tilde{\theta}'_k(t))+O(\e).
\end{eqnarray}
Using the three relations \eqref{eqn:nojoint-1},\eqref{eqn:nojoint-2} and \eqref{eqn:wft}, we have
\begin{eqnarray}
  \left|a_k(t)e^{-i\theta_k(t)}\widehat{\psi}(\omega\theta'_k(t)) - \tilde{a}_k(t)e^{-i\tilde{\theta}_k(t)}
\widehat{\psi}(\omega\tilde{\theta}'_k(t))\right|=O(\e),\quad \forall (t,\omega)\in U_{\theta_k}\bigcup U_{\tilde{\theta}_k},
\end{eqnarray}
which implies that
\begin{eqnarray}
\label{eqn:id-wft}
  \left|a_k(t)\widehat{\psi}(\omega\theta'_k(t))\right| - \left|\tilde{a}_k(t)
 \widehat{\psi}(\omega\tilde{\theta}'_k(t))\right|=O(\e),\quad \forall (t,\omega)\in U_{\theta_k}\bigcup U_{\tilde{\theta}_k}.
\end{eqnarray}
Next, we will prove that $|a_k(t)-\tilde{a}_k(t)|=O(\e)$ and $|\theta'_k(t)-\tilde{\theta}'_k(t)|=O(\e)$ from \eqref{eqn:id-wft}.
First, we consider the envelopes $a_k$ and $\tilde{a}_k$.

If $a_k(t)>\tilde{a}_k(t)$, we choose $\omega=1/\theta'_k(t)$ to get
\begin{eqnarray}
  a_k(t)\left|\widehat{\psi}(1)\right| - \tilde{a}_k(t)\left|
\widehat{\psi}(\tilde{\theta}'_k(t)/\theta'_k(t))\right|=O(\e) ,
\end{eqnarray}
where we have used the fact that $\left|\widehat{\psi}(\xi )\right|$ achieves
its maximum at $\xi =1$.
Since $a_k(t)>\tilde{a}_k(t)$, we have
\begin{eqnarray}
0\le \left|\widehat{\psi}(1)\right|\left( a_k(t) - \tilde{a}_k(t)\right)\le
 a_k(t)\left|\widehat{\psi}(1)\right| - \tilde{a}_k(t)\left|
\widehat{\psi}(\tilde{\theta}'_k(t)/\theta'_k(t))\right|=O(\e) .
\end{eqnarray}
This proves that
\begin{eqnarray}
  a_k(t) - \tilde{a}_k(t)=O(\e) .
\end{eqnarray}
If $a_k(t)<\tilde{a}_k(t)$, we take $\omega=1/\tilde{\theta}'_k(t)$. 
By following a similar argument, we can prove
\begin{eqnarray}
  \tilde{a}_k(t)-a_k(t)=O(\e) .
\end{eqnarray}
Combining these two cases, we obtain
\begin{eqnarray}
  |a_k(t) - \tilde{a}_k(t)|=O(\e) .
\end{eqnarray}
Substituting the above relation to \eqref{eqn:id-wft}, we get
\begin{eqnarray}
\label{eqn:dt-diff}
  \left|\widehat{\psi}(\omega\theta'_k(t))\right| - \left|\widehat{\psi}(\omega\tilde{\theta}'_k(t))\right|=O(\e),
\quad \forall (t,\omega)\in U_{\theta_k}\bigcup U_{\tilde{\theta}_k} .
\end{eqnarray}
For any $t\in [0,1]$, let $\omega=(1-\Delta/2)/\theta'_k(t)$, then we have
\begin{eqnarray}
  \left|\widehat{\psi}(1-\Delta/2))\right| - \left|\widehat{\psi}\left[(1-\Delta/2)\tilde{\theta}'_k(t)/\theta'_k(t)\right]\right|=O(\e).
\end{eqnarray}
Since $\widehat{\psi}$ is a fifth order B-Spline function and $\e \ll 1$, it is easy to see that there exists a constant $C>0$, such that
\begin{eqnarray}
\label{eqn:dt-est1}
  \left|(1-\Delta/2)\left(\tilde{\theta}'_k(t)/\theta'_k(t)-1\right)\right|\le C\e .
\end{eqnarray}
Since $\widehat{\psi}(\omega)$ is a fifth order B-Spline function,  
$\widehat{\psi}(\omega)$ is symmetric with respect to $\omega=1$. Thus we have 
$\widehat{\psi}(1-\Delta/2)=\widehat{\psi}(1+\Delta/2)$. Then, there is also another possibility:
\begin{eqnarray}
\label{eqn:dt-est2}
  \left|(1-\Delta/2)\tilde{\theta}'_k(t)/\theta'_k(t)-(1+\Delta/2)\right|\le C\e .
\end{eqnarray}
If \eqref{eqn:dt-est1} holds, then we get
\begin{eqnarray}
 \frac{ |\theta'_k(t)-\tilde{\theta}'_k(t)|}{\theta'_k(t)}=O(\e) .
\end{eqnarray}
If \eqref{eqn:dt-est2} holds, then we have
\begin{eqnarray}
  \frac{\tilde{\theta}'_k(t)}{\theta'_k(t)}\ge \frac{1+\Delta/2-C\e}{1-\Delta/2}\ge 1+\Delta ,
\end{eqnarray}
where we have used the assumption that $\e\ll 1$.

Then, let $\omega=1/\theta'_k(t)$ in \eqref{eqn:dt-diff}, we have
\begin{eqnarray}
1=  \left|\widehat{\psi}(1)\right| - \left|\widehat{\psi}\left(\tilde{\theta}'_k(t)/\theta'_k(t)\right)\right|=O(\e) .
\end{eqnarray}
This argument shows that \eqref{eqn:dt-est2} cannot be true. This completes the proof.
\end{proof}

\subsection{Discussion on Signals with close frequencies }

In this section, we will give a brief discussion on the signal with close frequencies, i.e. the frequency ratio $d\rightarrow 1$
in the definition of well-separated signal.

In the proof of Theorem \ref{thm:unique-intro}, the frequency ratio $d$ seems to be arbitrary, as long as it is larger than 1. Actually,
the distance between $d$ and 1 can not be too small. The gap is determined by the separation factor $\e$. First, it is easy to show that
the integrals $I_1, I_2, I_3$ in Lemma \ref{lem:wft} satisfy
\begin{eqnarray}
  I_1=O(\Delta^{-1}), \quad  I_2=O(\Delta^{-2}), \quad I_3=O(\Delta^{-3}).
\end{eqnarray}
When $d$ approaches 1, $\Delta=\frac{\sqrt{d}-1}{\sqrt{d}+1}$ becomes smaller, then $I_1, I_2, I_3$ become larger. When $I_1, I_2, I_3$ is as large as
$1/\e$, Lemma \ref{lem:wft} breaks down, which in turn leads to the break down of the proof of Theorem \ref{thm:unique-intro}.

On the other hand, the frequency ratio is defined pointwisely. As long as there exists a point, such that the frequency ratio at this point is away from 1
comparing with $\e$, the argument before \eqref{eqn:sparsity} of the proof of Theorem \ref{thm:unique-intro} still applies, which means that the decomposition given
in \eqref{opt01} is still an optimal solution of the optimization problem \eqref{P0}. But in this case, the uniqueness is not guaranteed. To illustrate this point, we construct
an example such that the solution of \eqref{P0} is not unique.
\begin{eqnarray}
\label{eqn:ex1}
  f(t)=\cos\theta_1(t)+\cos\theta_2(t),\quad t\in [0,1],
\end{eqnarray}
\begin{eqnarray}
  \theta_1(t)=6\pi k t +k\pi,\quad \theta_2(t)=8\pi kt+k\sin 2\pi t,
\end{eqnarray}
where $k$ is a positive integer. We can choose $k$ large enough, such that $\cos\theta_1,\cos\theta_2\in U_\e$.
At $t=0$, $\theta_2'(0)/\theta_1'(0)=4/3$. From the discussion above, we know that \eqref{eqn:ex1} gives a solution of \eqref{P0}.
On the other hand, it is easy to check that the following decomposition is also a solution of \eqref{P0}.
\begin{eqnarray}
\label{eqn:ex2}
  f(t)=\cos\phi_1(t)+\cos\phi_2(t),\quad t\in [0,1],
\end{eqnarray}
\begin{eqnarray*}
  \phi_1(t)=\left\{\begin{array}{cc} 6\pi k t +k\pi,& t\in [0,1/2],\\
8\pi kt+k\sin 2\pi t,& t\in (1/2,1],\end{array}
\right.
\quad \phi_2(t)=\left\{\begin{array}{cc} 8\pi kt+k\sin 2\pi t,& t\in [0,1/2],\\
6\pi k t +k\pi,& t\in (1/2,1].\end{array}
\right.
\end{eqnarray*}
This example shows that when the frequencies intersect, the solution of \eqref{P0} may not be unique. In this case, we need to
impose some extra constraints to obtain uniqueness of the decomposition. One natural idea is to pick up the solution according to the regularity of
$a(t)$ and $\theta'(t)$. The solution with smoother amplitude and frequency is favorable. One method based on this idea is proposed in \cite{HS14}
to decompose signals that do not have well-separated IMFs. In that method, the regularity is related with the sparsity over Fourier (or Wavelet) dictionary.
In the above example, this method would prefer the decomposition \eqref{eqn:ex1}, since in this decomposition, both of the amplitude and the
frequencies are very sparse over the Fourier dictionary.

\section{Optimal Solutions to the $P_2$ Problem}

Using the result in the previous section, we know that for signals that are well-separated, the solution of
\eqref{P0} is unique up to the separation factor $\e$. One natural question is that how to find this unique solution.

Based on matching pursuit, we propose the following algorithm to solve \eqref{P0} approximately.
\begin{itemize}
\item $r_0=f,\quad k=1$.
\item[Step 1:] Solve the following nonlinear least-square problem:
\begin{eqnarray}
\label{opt-greedy}
\begin{array}{rcl}\vspace{-2mm}
 (a_k,\theta_k)\in &\mbox{Argmin}&
\|r_{k-1}-a\cos\theta\|_{l^2}^2\\
&\scriptstyle a,\theta& \\
&\mbox{Subject to:}& a,\theta\in U_\e.
\end{array}
\end{eqnarray}
\item[Step 2:] Update the residual
\begin{eqnarray}
r_{k}=f-\sum_{j=1}^{k}a_{j}\cos\theta_{j}.
\end{eqnarray}
\item[Step 3:] If $\|r_{k}\|_{l^2}<\epsilon_0$, stop. Otherwise, set $k=k+1$ and go to Step 1.
\end{itemize}

In the above algorithm based on matching pursuit, each IMF is given by solving the following nonlinear least-square problem:
\begin{align}\label{P2}\tag{$P_2$}
\begin{array}{rcl}
 \mbox{Minimize} & & p(a, \theta):=\|f(t)-a(t)\cos\theta(t)\|_{L^2}^2 \;. \\
 \mbox{subject to} & & (a,\theta)\in U_\e.
\end{array}
\end{align}
We would like to know that under what conditions on each IMF $a_k(t)\cos \theta_k(t)$ of $f(t)$ the above nonlinear least-square problem could provide
a local (approximate) optimizer to the $P_0$ problem.

In the computations, we always deal with signals with finite time span. Without loss of generality, 
in this section we assume the time span of the signal is $[0,1]$. For the signal with finite length, it is unavoidable to introduce larger errors near the end points of the time interval. This is also known as the ``end effect''. To simplify the analysis, we assume that the signal is periodic over $[0,1]$. For those signals that are
 not periodic, we first 
multiply a cutoff function to make the signal vanish near the boundary and then treat it as a periodic signal. This means that 
the analysis in this section is valid only in the interior region away from the boundary.

For a periodic signal, we prove that each IMF is a local minimizer of \eqref{P2} as long as the signal satisfies some 
assumptions. This result is stated in Theorem \ref{thm:opt-l2}.  
\begin{theorem}\label{thm:opt-l2}
Let $f(t)$ be a function satisfying the scale-separation property with
separation factor $\e$ and frequency ratio $d$ as defined in
Definition \ref{def:well-sep-intro}:
  \begin{eqnarray*}
    f(t)=\sum_{k=1}^Ma_k(t)\cos\theta_k(t)+r(t), \quad a_k\cos\theta_k\in U_\e,\quad a_k=O(1),\; r=O(\e).
  \end{eqnarray*}
Suppose there exists $\alpha\in [1,d)$ and $l \in \{1,\cdots,M\}$ such that
\begin{align}\label{opt33}
 \alpha^{-1}\theta_{l}'(t)\leq \theta'(t)\leq \alpha\, \theta_{l}'(t),\qquad \forall t\in [0,1].
\end{align}
If
\begin{align}\label{opt34}
 p(a,\theta)\leq p(a_l, \theta_l),
\end{align}
where $p(a,\theta)$ is given in \eqref{P2},
then we have
\begin{align}\label{opt35}
 \frac{\|a\cos\theta-a_l\cos \theta_l\|_{L^2}}{\|a_l\cos \theta_l\|_{L^2}}=O(\sqrt{\e}) .
\end{align}
\end{theorem}

\begin{proof} 
First we know
\begin{align}\label{opt37}
\begin{split}
0&\geq p(a, \theta)-p(a_l, \theta_l) \\
&=\|f(t)-a(t)\cos\theta(t)\|_{L^2}^2-\|f(t)-a_l(t)\cos\theta_l(t)\|_{L^2}^2   \\
&=\left\|\sum_{k\neq l}a_{k}\cos\theta_{k}+r(t)+a_l(t)\cos\theta_l(t)-a(t)\cos\theta(t)\right\|_{L^2}^2-\left\|\sum_{k\neq l}a_{k}\cos\theta_{k}+r(t)\right\|_{L^2}^2 \\
&=\|a_l\cos\theta_l-a\cos\theta\|_{L^2}^2
+2\langle a_l\cos\theta_l-a\cos\theta, \sum_{k\neq l}a_{k}\cos\theta_{k}+r(t)\rangle,
\end{split}
\end{align}
where the first equality follows from the definition of $p(a,\theta)$ in \eqref{P2}.
In the rest of the proof, we try to control the second term of the above inequality.

It is easy to verify that
\begin{align}\label{opt38}
\begin{split}
\left| \langle a_l\cos\theta_l, \sum_{k\neq l}a_{k}\cos\theta_{k}\rangle \right|
&\leq \left[ \sum_{k\neq l} \mu_{k,l}  \|a_l\cos\theta_l\|_{L^2}\|a_k\cos\theta_k\|_{L^2} \right] \\
&=\delta_1 \|a_l\cos\theta_l\|^2,
\end{split}
\end{align}
where
\begin{align}\label{opt36a}
\mu_{k,l}=\frac{\left| \langle a_l\cos\theta_l, a_{k}\cos\theta_{k}\rangle \right|}{\|a_l\cos\theta_l\|_{L^2}\|a_k\cos\theta_k\|_{L^2}},\quad
\delta_1= \sum_{k\neq l} \mu_{k,l}\frac{\|a_k\cos \theta_k\|_{L^2}}{\|a_l\cos \theta_l\|_{L^2}} .
\end{align}
Similarly, we have
\begin{align}\label{opt41}
\begin{split}
&\left| \langle a\cos\theta, \sum_{k\neq l}a_{k}\cos\theta_{k}+r(t)\rangle \right| \\
< &\delta_2 \|a\cos\theta\|_{L^2}  \|a_l\cos\theta_l\|_{L^2} \\
\leq &\delta_2 \|a\cos\theta- a_l\cos\theta_l\|_{L^2}\cdot \|a_l\cos\theta_l\|_{L^2}
+\delta_2 \|a_l\cos\theta_l\|_{L^2}^2 ,
\end{split}
\end{align}
with
\begin{align}\label{opt36b}
\begin{split}
 &\delta_2= \sum_{k\neq l} \mu_{k,l,\alpha}\frac{\|a_k\cos \theta_k\|_{L^2}}{\|a_l\cos \theta_l\|_{L^2}}
 +\frac{\|r(t)\|_{L^2}}{\|a_l\cos \theta_l\|_{L^2}},\quad \\
&\mu_{k,l,\alpha}=\frac{\left| \langle a\cos\theta, a_{k}\cos\theta_{k}\rangle \right|}{\|a\cos\theta\|_{L^2}\|a_k\cos\theta_k\|_{L^2}} .
\end{split}
\end{align}
Thus it follows from \eqref{opt37}, \eqref{opt38} and  \eqref{opt41}
that
\begin{align}\label{opt42}
0> \|a_l\cos\theta_l-a\cos\theta\|_{L^2}^2 - 2\delta_2 \|a\cos\theta
- a_l\cos\theta_l\|_{L^2}\cdot \|a_l\cos\theta_l\|_{L^2}
-2(\delta_1+\delta_2)\|a_l\cos\theta_l\|_{L^2}^2,
\end{align}
which implies
\begin{align}
\label{opt35a}
 \frac{\|a\cos\theta-a_l\cos \theta_l\|_{L^2}}{\|a_l\cos \theta_l\|_{L^2}}
\leq \delta_2+\sqrt{\delta_2^2+2(\delta_1+\delta_2)}.
\end{align}
Here, $\mu_{k,l}$ and $\mu_{k,l,\alpha}$ are just the coherences between $a_l\cos\theta_l$, $a\cos\theta$ and $a_k\cos\theta_k$. Under the assumption of scale separation and the assumption that different IMFs are well separated, the behavior of $a_l\cos\theta_l$, $a\cos\theta$ and $a_k\cos\theta_k$ are close to that of the standard Fourier basis.
Then, it is natural to expect that the coherences, $\mu_{k,l}$ and $\mu_{k,l,\alpha}$, are small. Actually, the smallness of $\mu_{k,l}$ and $\mu_{k,l,\alpha}$
is given in Corollary \ref{crl01}. In particular, the estimate \eqref{opt43} from Corollary \ref{crl01} shows that for all $k\neq l$ we have
\begin{align}\label{opt51}
\begin{split}
\mu_{k,l}=O(\e),\qquad
\mu_{k,l,\alpha}=O(\e).
\end{split}
\end{align}
Together with the assumption that $r(t)=O(\e)$,  we get
\begin{align}\label{opt52}
\delta_1=O(\e),\qquad \delta_2=O(\e).
\end{align}
It follows from \eqref{opt35a} that 
\begin{align}\label{opt54}
 \frac{\|a\cos\theta-a_l\cos \theta_l\|_{L^2}}{\|a_l\cos \theta_l\|_{L^2}}= O(\sqrt{\e}).
\end{align}
This completes the proof.

\end{proof}

In the proof of the above theorem, we have used the estimate \eqref{opt51} for $\mu_{k,l}$ and $\mu_{k,l,\alpha}$. This estimate can be derived by the following lemma:
\begin{lemma}\label{lm02}
Let $(a, \theta)\in U_\e$ be such that $a\cos\theta$ has period $1$. Then we have
\begin{align}\label{opt07}
 \left(\frac{1}{2} -3\e \right) \|a(t)\|_{L^2}^2\le \|a(t)\cos\theta(t)\|_{L^2}^2\le \left(\frac{1}{2} +3\e \right) \|a(t)\|_{L^2}^2.
\end{align}
Furthermore, if there is another pair
$(\bar{a}, \bar{\theta})\in U_\e$ being periodic over $[0,1]$ such that
\begin{align}\label{opt09}
\beta:=\min_{t\in [0,1]} \frac{ \bar{\theta}'(t)}{ \theta'(t)} >1,
\end{align}
then we have
\begin{align}\label{opt08}
\begin{split}
\left| \langle a\cos\theta, \bar{a}\cos \bar{\theta}\rangle\right|
<4\e\left(1+\frac1{(1-\beta^{-1})^2}\right)\int_{0}^{1} a(t)\bar{a}(t)dt,
\end{split}
\end{align}
\end{lemma}
The proof of this lemma can be found in Appendix B. Then \eqref{opt51} is just a direct corollary of the above lemma.
\begin{corollary}\label{crl01}
Let $(a_k,\theta_k),\; k=1,\cdots,M$ be well-separated with frequency ratio $d$ and
separation factor $\e$ as defined in Definition \ref{def:well-sep-intro}. Let $(a,\theta)\in U_\e$ and
there exists $\alpha\in [1,d)$ and $l \in \{1,\cdots,M\}$ such that
\begin{align}
 \alpha^{-1}\theta_{l}'(t)\leq \theta'(t)\leq \alpha\, \theta_{l}'(t),\qquad \forall t\in [0,1].
\end{align}
Then for any $1\leq k\ne l\leq M$, we have
\begin{align}\label{opt43}
\begin{split}
\mu_{k,l}:=\frac{|\langle a_{k}\cos\theta_{k}, a_{l}\cos\theta_{l}\rangle|}{\|a_l\cos\theta_l\|_{L^2}\|a_k\cos\theta_k\|_{L^2}}
=O(\e), \\
\mu_{k,l,\alpha}:=\frac{|\langle a_{k}\cos\theta_{k}, a\cos\theta\rangle|}{\|a\cos\theta\|_{L^2}\|a_k\cos\theta_k\|_{L^2}} 
=O(\e).
\end{split}
\end{align}
\end{corollary}
\begin{proof}
For any $1\leq k\ne l\leq M$, suppose $\theta'_k>\theta'_l$. Since $(a_k,\theta_k),\; k=1,\cdots,M$ are well-separated,
we have
\begin{eqnarray}
  \min_{t\in [0,1]}\frac{\theta'_k}{\theta'_l}>d^{|k-l|}.
\end{eqnarray}
 Using Lemma \ref{lm02}, it is easy to check that
\begin{align*}
\begin{split}
\mu_{k,l}&=\frac{|\langle a_{k}\cos\theta_{k}, a_{l}\cos\theta_{l}\rangle|}{\langle a_k,a_l\rangle}\cdot
\frac{\langle a_k,a_l\rangle}{\|a_l\cos\theta_l\|_{L^2}\|a_k\cos\theta_k\|_{L^2}}\\
&\le\frac{|\langle a_{k}\cos\theta_{k}, a_{l}\cos\theta_{l}\rangle|}{\langle a_k,a_l\rangle}\cdot
\frac{\|a_k\|_{L^2}}{\|a_k\cos\theta_k\|_{L^2}}\cdot
\frac{\|a_l\|_{L^2}}{\|a_l\cos\theta_l\|_{L^2}}\\
&<4\e\left(\frac1{2}-3\e\right)^{-1}\left(1+\frac1{(1-d^{-|l-k|})^2}\right).\\
\end{split}
\end{align*}
For any $1\leq k\ne l\leq M$, since $(a_k,\theta_k),\; k=1,\cdots,M$ are well-separated and
\begin{align}
 \alpha^{-1}\theta_{l}'(t)\leq \theta'(t)\leq \alpha\, \theta_{l}'(t),\qquad \forall t\in [0,1].
\end{align}
if $\theta'_k<\theta'_l$, we know that
\begin{align}\label{opt39a}
\frac{ \theta'}{ \theta_k'}\geq \alpha^{-1}d^{l-k}.
\end{align}
And if $\theta'_k>\theta'_l$, we have
\begin{align}\label{opt40a}
\frac{\theta_k'}{ \theta'}\geq \alpha^{-1}d^{k-l}.
\end{align}
Then it follows from Lemma \ref{lm02} that, for $k\neq l$,
\begin{align*}
\begin{split}
\mu_{k,l,\alpha}
&\le\frac{|\langle a_{k}\cos\theta_{k}, a\cos\theta\rangle|}{\langle a_k,a\rangle}\cdot
\frac{\|a_k\|_{L^2}}{\|a_k\cos\theta_k\|_{L^2}}\cdot
\frac{\|a\|_{L^2}}{\|a\cos\theta\|_{L^2}}\\
&<4\e\left(\frac1{2}-3\e\right)^{-1}\left(1+\frac1{(1-\alpha d^{-|l-k|})^2}\right).
\end{split}
\end{align*}
\end{proof}

\begin{remark}\label{rmk41}
We remark that in Corollary \ref{crl01} and Theorem \ref{thm:opt-l2} the condition \eqref{opt33} could be replaced by
\begin{align}\label{opt33a}
\frac{\theta'(t)}{\theta_{l-1}'(t)} \geq \frac{d}{\alpha}, \ \
\frac{\theta_{l+1}'(t)}{\theta'(t)} \geq \frac{d}{\alpha},
\qquad \forall t\in [0,1].
\end{align}
\end{remark}

 If we put a stronger assumption on the separation ratio $d$ such that $d>M'^2$ (where $M'$ is defined in Definition \ref{scale-seperation}), 
then we obtain  the following theorem about the global minimizer of the $P_2$ problem.
 \begin{theorem}\label{thm:opt-l3}
 Let $f(t)$ be a function satisfying the scale-separation property with
 separation factor $\e$ and frequency ratio $d$ as defined in
 Definition \ref{def:well-sep-intro}:
   \begin{eqnarray*}
     f(t)=\sum_{k=1}^Ma_k(t)\cos\theta_k(t)+r(t), \quad a_k\cos\theta_k\in U_\e,\quad a_k=O(1),\; r=O(\e).
   \end{eqnarray*}
 Suppose further that $d>M'^2$, and that there exists $l\in \overline{1,M}$ such that
 \begin{align}\label{opt61}
  \|a_{k}\cos \theta_k\|< \|a_{l}\cos \theta_l\|,\qquad \forall k\in \overline{1,M}, k\neq l.
 \end{align}
 If
 \begin{align}\label{opt62}
  p(a,\theta)\leq p(a_l, \theta_l),
 \end{align}
 where $p(a,\theta)$ is given in \eqref{P2},
 then we have
 \begin{align}\label{opt63}
  \frac{\|a\cos\theta-a_l\cos \theta_l\|_{L^2}}{\|a_l\cos \theta_l\|_{L^2}}=O(\sqrt{\e}) .
 \end{align}
 \end{theorem}

 \begin{proof}
 First, we claim that, for each $k\in \overline{1,M}$,
 \begin{align}\label{opt65}
 p(a_k, \theta_k)=\|f\|_{L^2}^2-\|a_k\cos \theta_k\|_{L^2}^2+O(\e).
 \end{align}
 Granting this, it follows from \eqref{opt61} that
 \begin{align}\label{opt65-2}
 p(a_k, \theta_k)>p(a_l, \theta_l)\qquad \forall k\neq l.
 \end{align}
 To show \eqref{opt65}, we notice that
 \begin{align}\label{opt66}
 \begin{split}
 &p(a_k, \theta_k)=\|f-a_k\cos \theta_k\|_{L^2}^2 \\
 =&\|f\|_{L^2}^2+\|a_k\cos \theta_k\|_{L^2}^2-2\langle a_k\cos \theta_k+\sum_{k'\neq k}a_{k'}\cos\theta_{k'}+r(t), a_k\cos \theta_k\rangle \\
 =&\|f\|_{L^2}^2-\|a_k\cos \theta_k\|_{L^2}^2-2\langle \sum_{k'\neq k}a_{k'}\cos\theta_{k'}+r(t), a_k\cos \theta_k\rangle.
 \end{split}
 \end{align}
 By Lemma \ref{lm02}, we get
 \begin{align}\label{opt67}
 \langle a_{k'}\cos\theta_{k'}, a_k\cos \theta_k\rangle=O(\e)
 \end{align}
 for each $k'\neq k$. By the assumption in Definition \ref{def:well-sep-intro}, we have $r(t)=O(\e)$, which implies
 \begin{align}\label{opt68}
 \langle r(t), a_k\cos \theta_k\rangle=O(\e).
 \end{align}
 Substituting \eqref{opt67} and \eqref{opt68} into \eqref{opt66} proves \eqref{opt65}.

 Let $\alpha:=\frac{\sqrt{d}}{M'}$ (so $\alpha>1$) and let
 \begin{align}\label{opt69}
 &\alpha_k:=
  \left\{\begin{array}{lcl}
  \max_{t\in [0,1]}\frac{\theta'(t)}{\theta_k'(t)}, & & k=1 \\
  \max_{t\in [0,1]}\frac{\theta_k'(t)}{\theta'(t)}, & & k=M \\
  \max\left\{ \max_{t\in [0,1]}\frac{\theta'(t)}{\theta_k'(t)}, \max_{t\in [0,1]}\frac{\theta_k'(t)}{\theta'(t)} \right\}, & &  1<k<M
  \end{array}\right.
   \\ \label{opt70}
 &\alpha_{min}:= \min_{1\le k\le M}\left\{ \alpha_k \right\}.
 \end{align}
 Let $k_0$ be such that $\alpha_{k_0}=\alpha_{min}$. We Claim that
 \begin{enumerate}
 \item[(1)] If $k_0>1$,
 \begin{align}\label{opt71}
 \frac{\theta'(t)}{\theta_{k_0-1}'(t)}\ge \alpha, \qquad \forall t\in [0,1].
 \end{align}

 \item[(2)] If $k_0<K$,
 \begin{align}\label{opt72}
 \frac{\theta_{k_0+1}'(t)}{\theta'(t)}\ge \alpha, \qquad \forall t\in [0,1].
 \end{align}
 \end{enumerate}
 We prove item (1) by contradiction. Assume that there exists $t_1\in [0,1]$ such that
 \begin{align}\label{opt71a}
 \frac{\theta'(t_1)}{\theta_{k_0-1}'(t_1)}< \alpha.
 \end{align}
 Then for any $t\in [0,1]$, we have
 \begin{align}\label{opt73}
 \frac{\theta'(t)}{\theta_{k_0-1}'(t)}=\frac{\theta'(t)}{\theta'(t_1)}\cdot \frac{\theta'(t_1)}{\theta_{k_0-1}'(t_1)}
 \cdot \frac{\theta_{k_0-1}'(t_1)}{\theta_{k_0-1}'(t)}< M'^2\alpha=M'\sqrt{d}.
 \end{align}
 On the other hand, \eqref{opt71a} also implies
 that
 \begin{align}\label{opt74}
 \frac{\theta_{k_0}'(t_1)}{\theta'(t_1)}
 =\frac{\theta_{k_0}'(t_1)}{\theta_{k_0-1}'(t_1)} \left( \frac{\theta'(t_1)}{\theta_{k_0-1}'(t_1)}\right)^{-1}> \frac{d}{\alpha}=M'\sqrt{d}.
 \end{align}
 Using \eqref{opt73}, \eqref{opt74}, and the following estimates
 \begin{align}\label{opt75-1}
 \frac{\theta_{k_0-1}'(t_1)}{\theta'(t_1)}
 \leq
 \frac1{d}\cdot\frac{\theta_{k_0}'(t_1)}{\theta'(t_1)}<\frac{\theta_{k_0}'(t_1)}{\theta'(t_1)},
 \end{align}
 we get
 \begin{align}\label{opt75}
 \alpha_{k_0-1}
 <
 \max_{t\in [0,1]} \frac{\theta_{k_0}'(t)}{\theta'(t)}\leq \alpha_{k_0}.
 \end{align}
 This contradicts $\alpha_{k_0}=\alpha_{min}$. So item (1) is satisfied. Similarly we can show that item (2) is true.

 It follows from \eqref{opt62} and \eqref{opt65-2} that
 \begin{align}\label{opt76}
  p(a,\theta)\leq p(a_{k_0}, \theta_{k_0}).
 \end{align}
 Using Theorem \ref{thm:opt-l2} by replacing condition \eqref{opt33} by \eqref{opt33a} in Remark \ref{rmk41}, we obtain that
 \begin{align}\label{opt77}
  \frac{\|a\cos\theta-a_{k_0}\cos \theta_{k_0}\|_{L^2}}{\|a_{k_0}\cos \theta_{k_0}\|_{L^2}}= O(\sqrt{\e}).
 \end{align}
 This implies
 \begin{align}\label{opt78}
 \|a\cos\theta-a_{k_0}\cos \theta_{k_0}\|_{L^2} = O(\sqrt{\e}).
 \end{align}
 Therefore we obtain
 \begin{align}\label{opt79}
 \begin{split}
 &\left|p(a, \theta)-p(a_{k_0}, \theta_{k_0})\right| \\
 =&\|a_{k_0}\cos\theta_{k_0}-a\cos\theta\|_{L^2}^2
 +2\left| \langle a_{k_0}\cos\theta_{k_0}-a\cos\theta, \sum_{k\neq {k_0}}a_{k}\cos\theta_{k}+r(t)\rangle\right| \\
 =& O(\sqrt{\e}),
 \end{split}
 \end{align}
 where the first equality is deduced using an argument similar to the equalities in \eqref{opt37}, and the second one follows from \eqref{opt78}.
 So by \eqref{opt62}, we see that $k_0=l$ and $\alpha_k >\alpha_l$ for all $k\neq l$. Thus \eqref{opt77} implies \eqref{opt63}. This completes the proof.

 \end{proof}

The condition that $d>M'^2$ is very strong when $M'$ is big. However, if this condition is violated, the global minimizer to the $P_2$ problem
may not be any of the IMFs in the sparsest representation. For example, let us consider an artificial signal $f(t)$ defined as follows:
\begin{align}\label{opt81}
\begin{split}
f(t)=&a_1(t)\cos\theta_1(t)+a_2(t)\cos\theta_2(t), \quad t\in [0,6], \\
\text{where } & a_1(t)=2+t, \ \ \theta_1(t)=\left\{\begin{array}{lcl}
10\pi t & & t\in [0,2] \\ 20\pi + 10\pi(t-2) +\frac{5\pi}{3}(t-2)^3  & & t\in [2,3] \\ 
50\pi+20\pi(t-4)-\frac{5\pi}{3}(t-4)^3 & & t\in [3,4] \\ 
50\pi+20\pi(t-4) & & t\in [4,6]
\end{array}\right. , \\
&a_2(t)=8-t, \ \  \theta_2(t)=2\theta_1(t).
\end{split}
\end{align}
Here we see that $d=M'=2$. Moreover, it is easy to verify that
\begin{align}\label{opt82}
\left|\frac{a_1'(t)}{\theta_1'(t)}\right|\le \frac1{10\pi},\; \;
\left|\frac{\theta_1''(t)}{\left(\theta_1'(t)\right)^2}\right|\le\frac1{10\pi};\; \;
\left|\frac{a_2'(t)}{\theta_2'(t)}\right|\le \frac1{20\pi},\; \;
\left|\frac{\theta_2''(t)}{\left(\theta_2'(t)\right)^2}\right|\le\frac1{20\pi}, \quad \forall t\in [0,6].
\end{align}
So $a_1(t)\cos\theta_1(t)$ and $a_2(t)\cos\theta_2(t)$ are well separated and both satisfy the scale separation property on the time domain $[0,6]$.
If we attempt to solve the problem \eqref{P2} with parameter $\e=\frac1{10\pi}$, we will find that the solution $(a,\theta)$ with
\begin{align}\label{opt83}
a(t)=\max\{a_1(t), a_2(t)\}=5+|t-3|, \ \ \theta(t)=20\pi t
\end{align}
is better than $(a_1,\theta_1)$ or $(a_2,\theta_2)$ (here we note that $(a, \theta)$ is obtained by simply connecting the left
part of $(a_2, \theta_2)$ and the right part of $(a_1, \theta_1)$).
In fact, numerical estimation gives us that
\begin{align}\label{opt84}
\begin{split}
&p(a, \theta)=\|f(t)-a(t)\cos\theta(t)\|_{L^2}\approx 72.4, \\
&p(a_1, \theta_1)=\|a_2\cos\theta_2\|_{L^2}^2\approx 84, \\
&p(a_2, \theta_2)=\|a_1\cos\theta_1\|_{L^2}^2\approx 84.
\end{split}
\end{align}
Thus if $d\le M'^2$, mode mixing may occur. 

However, Theorem \ref{thm:opt-l3} is still helpful for us even if the condition that $d>M'^2$ is violated. 
Consider the signal $\displaystyle f(t)=\sum_{k=1}^Ma_k(t)\cos\theta_k(t)+r(t)$ on $[0,1]$.
We could partition the whole time domain $[0,1]$ into $m$ subintervals $[t_{i-1}, t_{i}] (i=1,\dots, m)$, where
$0=t_0<t_1<\dots<t_m=1$, and on each subinterval, we have 
\begin{align}\label{opt85}
\frac{\sup_{t\in [t_{i-1}, t_{i}]}\theta_k'(t)}{\inf_{t\in [t_{i-1}, t_{i}]}\theta_k'(t)}<\sqrt{d}, 
\quad k=1,2,\dots, M.
\end{align}
So we could find each IMF on each subinterval efficiently, and connect different parts of each IMF together. 
This is the basic idea of the upcoming work of the authors: Two-level method in sparse time-frequency
representation of multiscale data.

\section{Concluding Remarks}
\label{sec:conclude}

In this paper, we discussed the uniqueness of the decomposition obtained by the sparse time frequency decomposition.
We proved that under the assumption of scale separation, the decomposition is unique up to an error associate
with the scale separation. Moreover, we showed that under the same assumption, nonlinear matching pursuit could be used to obtain this
sparse decomposition. The results in this paper establish a solid foundation for the sparse time frequency
decomposition.

In our future work, we would like to relax some of the assumptions of scale separation. In many problems, the decompositions seem to be unique although the scale separation is not satisfied. We plan to perform further theoretical study and get some guidance to decompose signals with poor scale separation.

\vspace{0.2in}
\noindent
{\bf Acknowledgments.}
This work was supported by NSF FRG Grant DMS-1159138, DMS-1318377, an AFOSR MURI Grant FA9550-09-1-0613 and a DOE grant DE-FG02-06ER25727.
The research of Dr. Z. Shi was also in part supported by a NSFC Grant 11201257.
The research of Dr. C.G. Liu was also supported by a NSFC Grant 11371173.

\appendix
\setcounter{section}{1}
\vspace{10mm}
\noindent\textbf{\large  Appendix A: Proof of Lemma \ref{lem:wft}}
\begin{proof} {\it of Lemma \ref{lem:wft}}

First, we have that
\begin{eqnarray}
  \widehat{\psi}(\omega\theta'(t))=\int_{\mathbb{R}} e^{-i\omega\theta'(t)z}\psi(z)\mathd z=\frac{1}{\omega}\int_{\mathbb{R}} e^{-i\theta'(t)z}\psi\left(\frac{z}{\omega}\right)\mathd z .
\end{eqnarray}
Using this relation, we obtain
\begin{eqnarray}
&&  \frac{1}{\sqrt{\omega}}\int_{\mathbb{R}} a(\tau)e^{-i\theta(\tau)}\psi\left(\frac{\tau-t}{\omega}\right)\mathd \tau-|\omega|^{1/2}a(t)e^{-i\theta(t)}\widehat{\psi}(\omega\theta'(t))\nonumber\\
&=& \frac{1}{\sqrt{\omega}}\left[\int_{\mathbb{R}} (a(\tau)-a(t))e^{-i\theta(\tau)}\psi\left(\frac{\tau-t}{\omega}\right)\mathd \tau+
a(t)\int_{\mathbb{R}} (e^{-i\theta(\tau)}-e^{-i(\theta(t)+\theta'(t)(\tau-t))})\psi\left(\frac{\tau-t}{\omega}\right)\mathd \tau \right]. \nonumber\\
\label{eqn:error-wft}
\end{eqnarray}
For the first term, we have
\begin{eqnarray}
&&  |\omega|^{-1/2}\int_{\mathbb{R}} (a(\tau)-a(t))e^{-i\theta(\tau)}\psi\left(\frac{\tau-t}{\omega}\right)\mathd \tau\nonumber\\
&=& |\omega|^{-1/2}\int_{\mathbb{R}} h(\tau,t)e^{-i\theta(\tau)}\mathd \tau\nonumber\\
&=& -i |\omega|^{-1/2}\int_{\mathbb{R}} \left(\frac{h(\tau,t)}{\theta'(\tau)}\right)'e^{-i\theta(\tau)}\mathd \tau , \nonumber
\end{eqnarray}
where
\begin{eqnarray}
   h(\tau,t)=(a(\tau)-a(t))\psi\left(\frac{\tau-t}{\omega}\right) .
\end{eqnarray}
Direct calculation gives that
\begin{eqnarray}
  \left(\frac{h(\tau,t)}{\theta'(\tau)}\right)'=\frac{a'(\tau)}{\theta'(\tau)}\psi\left(\frac{\tau-t}{\omega}\right)
+\frac{a(\tau)-a(t)}{\theta'(\tau)}\frac1{\omega}\psi'\left(\frac{\tau-t}{\omega}\right)
-\frac{h(\tau,t)\theta''(\tau)}{(\theta'(\tau))^2} .
\end{eqnarray}
Using the assumption that $(a,\theta)\in U_\e$, we have
\begin{eqnarray}
  \frac{a'(\tau)}{\theta'(\tau)}\le \e\quad \frac{\theta''(\tau)}{(\theta'(\tau))^2}\le \e ,
\end{eqnarray}
and
\begin{eqnarray}
\left|\frac{a(\tau)-a(t)}{\theta'(\tau)}\right|=\left|\frac{a'(t_\tau)(\tau-t)}{\theta'(\tau)}\right|
=\left|\frac{\theta'(t_\tau)}{\theta'(\tau)}\cdot\frac{a'(t_\tau)(\tau-t)}{\theta'(t_\tau)}\right|
\leq M'\e|\tau-t|,
\end{eqnarray}
where $t_\tau$ is a point between $t$ and $\tau$.
Then, we obtain
\begin{eqnarray}
\label{eqn:a}
  |\omega|^{-1/2}\left|\int_{\mathbb{R}} (a(\tau)-a(t))e^{-i\theta(\tau)}\psi\left(\frac{\tau-t}{\omega}\right)\mathd \tau\right|
\le \e |\omega|^{1/2}\left[(A+|a(t)|+1)I_1+M'I_2\right],
\end{eqnarray}
where $A=\sup_{t\in\mathbb{R}}|a(t)|$ and
\begin{eqnarray}
  I_1=\int_{\mathbb{R}} |\psi(\tau)|\mathd \tau,\quad I_2=\int_{\mathbb{R}}|\tau\psi'(\tau)|\mathd \tau .
\end{eqnarray}
Now, we turn to bound the second term of \eqref{eqn:error-wft}. We have
\begin{eqnarray}
   &&|\omega|^{-1/2}
\int_{\mathbb{R}} (e^{-i\theta(\tau)}-e^{-i(\theta(t)+\theta'(t)(\tau-t))})\psi\left(\frac{\tau-t}{\omega}\right)\mathd \tau \nonumber\\
&=& |\omega|^{-1/2}\int_{\mathbb{R}} g(\tau,t)e^{-i\theta(\tau)}\mathd \tau\nonumber\\
&=& -i |\omega|^{-1/2}\int_{\mathbb{R}} \left(\frac{g(\tau,t)}{\theta'(\tau)}\right)'e^{-i\theta(\tau)}\mathd \tau , \nonumber
\end{eqnarray}
where
\begin{eqnarray}
  g(\tau,t)=(1-e^{i\Delta \theta})\psi\left(\frac{\tau-t}{\omega}\right),
\end{eqnarray}
and $\Delta\theta=\theta(\tau)-\theta(t)-\theta'(t)(\tau-t)$.

Direct calculations show that
\begin{eqnarray}
  \left(\frac{g(\tau,t)}{\theta'(\tau)}\right)'&=&-\frac{i(\theta'(\tau)-\theta'(t))}{\theta'(\tau)}e^{i\Delta\theta(\tau,t)}\psi\left(\frac{\tau-t}{\omega}\right)
+\frac{1}{\theta'(\tau)}(1-e^{i\Delta\theta(\tau,t)})\frac{1}{\omega}\psi'\left(\frac{\tau-t}{\omega}\right)\nonumber\\
&&-\frac{g(\tau,t)\theta''(\tau)}{(\theta'(\tau))^2} .
\label{eqn:dg}
\end{eqnarray}

Among the three terms, the third one is easiest to bound. We have
\begin{eqnarray}
\label{eqn:dg3}
  \left|\frac{g(\tau,t)\theta''(\tau)}{(\theta'(\tau))^2}\right|\le 2\e  \left|\psi\left(\frac{\tau-t}{\omega}\right)\right|.
\end{eqnarray}
It follows that
\begin{eqnarray}
  \label{eqn:g3}
  \left|\int_\mathbb{R}\frac{g(\tau,t)\theta''(\tau)}{(\theta'(\tau))^2}e^{-i\theta(\tau)}\mathd\tau\right|\le 2\e\int_\mathbb{R}|\psi(\tau)|\mathd\tau=2 I_1\e .
\end{eqnarray}
To bound the other two terms, we need some preparations.

First, by using the assumption $\left|\left(\frac{1}{\theta'}\right)'\right|
=\left|\frac{\theta''}{(\theta')^2}\right|\le \e$,
 we can bound $|\theta'(t)-\theta'(\tau)|$ by integrating the above inequality from $t$ to $\tau$ as follows
\begin{eqnarray}
  \left|\frac{1}{\theta'(\tau)}-\frac{1}{\theta'(t)}\right| \le \e |\tau-t|,
\end{eqnarray}
which gives that
\begin{eqnarray}
\label{eqn:t-diff}
  |\theta'(\tau)-\theta'(t)| \le \e \theta'(\tau)\theta'(t)|\tau-t|,\quad \forall t,\tau \in \mathbb{R}.
\end{eqnarray}

We now estimate $\left|1-e^{i\Delta\theta(\tau,t)}\right|$ in two different ways,
\begin{eqnarray}
\label{eqn:t-diff-1}
  \left|e^{-i\Delta\theta(\tau,t)}-1\right|&\le& |\Delta\theta(\tau,t)|\le |\theta'(t^*)-\theta'(t)||t-\tau|\le
\e \theta'(t)\theta'(t^*)|t-\tau|^2\nonumber\\ 
&\le& M'\e\theta'(\tau)\theta'(t)|t-\tau|^2,
\end{eqnarray}
where $t^*$ is a number between $t$ and $\tau$. And also
\begin{eqnarray}
\label{eqn:t-diff-2}
  \left|e^{-i\Delta\theta(\tau,t)}-1\right|&\le& |\Delta\theta(\tau,t)|\le |\theta'(t^*)-\theta'(t)||t-\tau|\le  
\left(\sup_{t\in\mathbb{R}}\theta'(t)-\inf_{t\in\mathbb{R}}\theta'(t)\right)|t-\tau|\nonumber\\
&\le&(M'-1)\left(\inf_{t\in\mathbb{R}}\theta'(t)\right)|t-\tau|
\le (M'-1)\theta'(t)|t-\tau|.
\end{eqnarray}
The first term can be estimated as follows,
\begin{eqnarray}
  \label{eqn:g1}
  &&\left|\int_\mathbb{R}\frac{i(\theta'(\tau)-\theta'(t))}{\theta'(\tau)}e^{i\Delta\theta(\tau,t)}
\psi\left(\frac{\tau-t}{\omega}\right)
e^{-i\theta(\tau)}\mathd \tau\right|\nonumber\\
&=&\left|\int_\mathbb{R}\left(1-\frac{\theta'(t)}{\theta'(\tau)}\right)
\psi\left(\frac{\tau-t}{\omega}\right)e^{i\theta'(t)(\tau-t)}\mathd \tau\right|\nonumber\\
&=&\left|\frac{1}{\theta'(t)}\int_\mathbb{R}\left[\left(1-\frac{\theta'(t)}{\theta'(\tau)}\right)\psi\left(\frac{\tau-t}{\omega}\right)\right]'e^{i\theta'(t)(\tau-t)}\mathd \tau\right|\nonumber\\
&\le&\left|
\int_\mathbb{R}\frac{\theta''(\tau)}{(\theta'(\tau))^2}\psi\left(\frac{\tau-t}{\omega}\right)e^{i\theta'(t)(\tau-t)}
\mathd \tau\right|+\left|\frac{1}{\omega}\int_\mathbb{R}\frac{\theta'(\tau)-\theta'(t)}{\theta'(t)\theta'(\tau)}
\psi'\left(\frac{\tau-t}{\omega}\right)e^{i\theta'(t)(\tau-t)}
\mathd \tau\right|\nonumber\\
&\le &\e |\omega|\left(\int_{\mathbb{R}}|\psi(\tau)|\mathd \tau+  \int_\mathbb{R}|\tau\psi'(\tau)|\mathd\tau\right)\nonumber\\
&\le & |\omega| (I_1+I_2)\e \;.
\end{eqnarray}
The second equality is obtained by integration by parts. To get the second inequality, we use the assumptions that 
$\left|\frac{\theta''(\tau)}{(\theta'(\tau))^2}\right|\le \e$ and \eqref{eqn:t-diff}.

Now, we turn to estimate the second term. 
Let $g_2(\tau,t)=\frac{1}{\theta'(\tau)}(1-e^{i\Delta\theta(\tau,t)})\frac{1}{\omega}\psi'\left(\frac{\tau-t}{\omega}\right)$.
First, using the fact that $\Delta\theta(\tau,t)=\theta(\tau)-\theta(t)-\theta'(t)(\tau-t)$, we have
\begin{eqnarray}
  \left|\int_{\mathbb{R}} g_2(\tau,t)e^{-i\theta(\tau)}\mathd \tau\right|&=& \left|e^{-i\theta(t)}\int_{\mathbb{R}} 
\frac{e^{-i\Delta\theta(\tau,t)}-1}{\omega\theta'(\tau)}\psi'\left(\frac{\tau-t}{\omega}\right)
e^{-i\theta'(t)(\tau-t)}\mathd \tau \right|\nonumber\\
&=&\left|\frac{1}{\theta'(t)}\int_{\mathbb{R}} 
\left[\frac{e^{-i\Delta\theta(\tau,t)}-1}{\omega\theta'(\tau)}\psi'\left(\frac{\tau-t}{\omega}\right)\right]'
e^{-i\theta'(t)(\tau-t)}\mathd \tau \right|,
\label{eqn:dg2}
\end{eqnarray}
and
\begin{eqnarray}
&&  \frac{1}{\theta'(t)}\left[\frac{e^{-i\Delta\theta(\tau,t)}-1}{\omega\theta'(\tau)}\psi'\left(\frac{\tau-t}{\omega}\right)\right]'\nonumber\\
&=&-\frac{\theta''(\tau)}{(\theta'(\tau))^2}\frac{(e^{-i\Delta\theta(\tau,t)}-1)}{\omega\theta'(t)}
\psi'\left(\frac{\tau-t}{\omega}\right)
-\frac{i(\theta'(\tau)-\theta'(t))}{\omega\theta'(\tau)\theta'(t)}e^{-i\Delta\theta(\tau,t)}
\psi'\left(\frac{\tau-t}{\omega}\right)\nonumber\\
&&+\frac{e^{-i\Delta\theta(\tau,t)}-1}{\omega^2\theta'(\tau)\theta'(t)}\psi''\left(\frac{\tau-t}{\omega}\right).
\end{eqnarray}

Then, \eqref{eqn:dg2} can be bounded term by term as follows: 
\begin{eqnarray}
&&\text{using \eqref{eqn:t-diff-2}},\quad
\left|\frac{\theta''(\tau)}{(\theta'(\tau))^2}\frac{(e^{-i\Delta\theta(\tau,t)}-1)}{\omega\theta'(t)}
\psi'\left(\frac{\tau-t}{\omega}\right)\right|
\le (M'-1)\e \left|\frac{\tau-t}{\omega}\psi'\left(\frac{\tau-t}{\omega}\right)\right|,\quad\quad\quad
\label{eqn:g2-1}
\\
&&\text{using \eqref{eqn:t-diff}},\quad
\left|  \frac{i(\theta'(\tau)-\theta'(t))}{\omega\theta'(\tau)\theta'(t)}e^{-i\Delta\theta(\tau,t)}
\psi'\left(\frac{\tau-t}{\omega}\right)
\right|
\le \e \left|\frac{\tau-t}{\omega}\psi'\left(\frac{\tau-t}{\omega}\right)\right|,
\label{eqn:g2-2}
\\
&&\text{using \eqref{eqn:t-diff-1}},\quad
  \left|\frac{e^{-i\Delta\theta(\tau,t)}-1}{\omega^2\theta'(\tau)\theta'(t)}\psi''\left(\frac{\tau-t}{\omega}\right)\right|
\le M' \e\left|\left(\frac{\tau-t}{\omega}\right)^2\psi''\left(\frac{\tau-t}{\omega}\right)\right| .
\label{eqn:g2-3}
\end{eqnarray}
By combining these inequalities, \eqref{eqn:g2-1},\eqref{eqn:g2-2} and \eqref{eqn:g2-3}, we get
\begin{eqnarray}
\label{eqn:g2}
  \left|\int_{\mathbb{R}} g_2(\tau,t)e^{-i\theta(\tau)}\mathd \tau\right|\le M' |\omega| (I_2+I_3)\e ,
\end{eqnarray}
where
\begin{eqnarray}
  I_3=\int_{\mathbb{R}}|\tau^2\psi''(\tau)|\mathd \tau .
\end{eqnarray}

\noindent
Then the proof is completed by combining \eqref{eqn:a}, \eqref{eqn:g3}, \eqref{eqn:g1} and \eqref{eqn:g2}.
\end{proof}

\vspace{10mm}
\noindent
\textbf{\large Appendix B: Proof of Lemma \ref{lm02}}

\vspace{3mm}

\setcounter{section}{2}
\setcounter{equation}{0}
To prove Lemma \ref{lm02}, we need the following technical lemma.
\begin{lemma}\label{lm01}
 Let $\e \in (0,1)$ and $g(t)$ be a positive, continuous, and piecewise $C^1$ function on $[c,c+2n\pi]$, where $n$ is an integer. Suppose
\begin{align}\label{smth11}
 \left|\frac{g'(t)}{g(t)}\right|<\e,\qquad \forall t\in [c,d].
\end{align}
Then we have
\begin{align}\label{smth12}
 \left|\int_c^{c+2n\pi} g(t)\cos t dt\right|
 <2\pi \e  \int_c^{c+2n\pi} g(t)dt.
\end{align}
\end{lemma}

\begin{proof}
For each $m\in \{0,1,2,\cdots, n\}$, let $t_m=c+2m\pi$. We have
\begin{align*}
\left| \int_c^{c+2n\pi} g(t)\cos t dt\right|&=\left| \sum_{m=1}^{\eta}\int_{t_{m-1}}^{t_m} g(t)\cos t dt\right|
=\left| \sum_{m=1}^{\eta}\int_{t_{m-1}}^{t_m} [g(t)-g(t_{m-1})]\cos t dt\right| \\
&=\left| \sum_{m=1}^{\eta} \int_{t_{m-1}}^{t_m}  \left( \int_{t_{m-1}}^{t} g'(s)ds \right) \cos t dt \right| \\
&\leq \sum_{m=1}^{\eta} \int_{t_{m-1}}^{t_m} |\cos t|\left( \int_{t_{m-1}}^{t} \left|g'(s)\right|ds \right) dt \\
&\leq \sum_{m=1}^{\eta} \left(\int_{t_{m-1}}^{t_m} |\cos t| dt\right)\left( \int_{t_{m-1}}^{t_m} \e g(s) ds \right)  \\
&=\sum_{m=1}^{\eta} 4 \e \left( \int_{t_{m-1}}^{t_m} g(s) ds \right) \\
&=4 \e  \int_{c}^{c+2n\pi} g(s)ds .
\end{align*}
We complete the proof.
\end{proof}


Now, we can give the proof of Lemma \ref{lm02}.
\begin{proof} {\it of Lemma \ref{lm02}}

First, using $\cos^2 (\theta) = (1+\cos(2\theta))/2$, we get
\begin{align}\label{opt11}
\begin{split}
\|a(t)\cos\theta(t)\|_{L^2}^2  =\frac1{2} \|a(t)\|_{L^2}^2+\frac1{2}\int_{0}^{1} a^2(t) \cos 2\theta(t) dt .
\end{split}
\end{align}
Let $s=2\theta(t)$. Then we obtain
\begin{align}\label{opt12}
\int_{0}^{1} a^2(t) \cos 2\theta(t) dt =\frac1{2}\int_{2\theta(0)}^{2\theta(1)} g(s) \cos s ds,
\end{align}
where $t(s):=\theta^{-1}(\frac{s}{2})$ and
\begin{align}\label{b12a}
 g(s):= \frac{a^2(t(s))}{\theta'(t(s))}.
\end{align}
So the derivative of $g$ is
\begin{align}\label{opt13}
g'(s)=\frac{a(t(s))a'(t(s))}{[\theta'(t(s))]^2}
-\frac{a^2(t(s))\cdot \theta''(t(s))}{2[\theta'(t(s))]^3}.
\end{align}
Hence we get
\begin{align}\label{opt14}
\left|\frac{g'(s)}{g(s)}\right|
=\left|\frac{a'(t(s))}{a(t(s))\cdot \theta'(t(s))} -\frac{\theta''(t(s))}{2[\theta'(t(s))]^2}\right|
< \frac{3}{2}\e,\qquad \forall s\in[2\theta(0), 2\theta(1)].
\end{align}
Using Lemma \ref{lm01}, we obtain
\begin{align}\label{opt15}
\begin{split}
\left| \int_{0}^{1} a^2(t) \cos 2\theta(t) dt \right|
&< 3\e\int_{2\theta(0)}^{2\theta(1)} g(s)ds
=6 \e \int_{0}^{1} a^2(t) dt.
\end{split}
\end{align}
The above estimate and \eqref{opt11} imply \eqref{opt07}.

To prove \eqref{opt08}, we represent the inner product in this inequality as follows:
\begin{align}\label{opt16}
\begin{split}
\langle a\cos\theta, \bar{a}\cos \bar{\theta}\rangle
=\frac1{2}\left[ \int_{0}^{1} a(t)\bar{a}(t)\cos (\bar{\theta}(t)+\theta(t)) dt
+\int_{0}^{1} a(t)\bar{a}(t)\cos (\bar{\theta}(t)-\theta(t)) dt\right] .
\end{split}
\end{align}
Let $s=\theta(t)-\hat{\theta}(t)$. We obtain
\begin{align}\label{opt17}
\int_{0}^{1} a(t)\bar{a}(t)\cos (\bar{\theta}(t)-\theta(t)) dt
=\int_{\bar{\theta}(0)-\theta(0)}^{\bar{\theta}(1)-\theta(1)} g(s)\cos s ds,
\end{align}
where $t(s)=(\bar{\theta}-\theta)^{-1}(s)$ and
\begin{align}\label{opt18}
 g(s)=\frac{a(t(s))\bar{a}(t(s))}{\bar{\theta}'(t(s))-\theta'(t(s))}.
\end{align}
Thus, we have
\begin{align}\label{opt19}
 \frac{d}{ds}g(s)=\frac{(\bar{\theta}'-\theta')(a\bar{a}'+a'\bar{a})-a\bar{a}(\bar{\theta}''-\theta'')}{\left(\bar{\theta}'(t(s))-\theta'(t(s))\right)^3}
\end{align}
and
\begin{align}\label{opt20}
\left|\frac{g'(s)}{g(s)}\right|
< \frac{2\e}{(1-\beta^{-1})^2},\qquad \forall s\in[\bar{\theta}(0)-\theta(0), \bar{\theta}(1)-\theta(1)].
\end{align}
Using Lemma \ref{lm01} again, we get
\begin{align}\label{opt21}
\begin{split}
\left| \int_{0}^{1} a(t)\bar{a}(t)\cos (\bar{\theta}(t)-\theta(t)) dt \right| 
<&\frac{8 \e}{(1-\beta^{-1})^2}\int_{0}^{1} a(t)\bar{a}(t)dt.
\end{split}
\end{align}
Similarly, we can show
\begin{eqnarray}\label{opt22}
\left| \int_{0}^{1} a(t)\bar{a}(t)\cos (\bar{\theta}(t)+\theta(t)) dt \right|< \frac{8 \e(\beta^2+\beta+1)}{(1+\beta)^2}\int_{0}^{1} a(t)\bar{a}(t)dt
<8\e\int_{0}^{1} a(t)\bar{a}(t)dt
\end{eqnarray}
Thus \eqref{opt08} follows by combining \eqref{opt16}, \eqref{opt21} and \eqref{opt22}.
\end{proof}

\bibliographystyle{plain}
\bibliography{EMD}

\end{document}